\documentclass[sigplan,screen,nonacm]{acmart}
\usepackage[utf8]{inputenc}
\usepackage{hyperref}
\usepackage{algorithmic}
\usepackage{graphicx}
\usepackage{textcomp}
\usepackage{xcolor}
\usepackage{soul}
\usepackage{algorithm}
\usepackage{fancyvrb}
\usepackage{pifont}
\usepackage{subcaption}
\VerbatimFootnotes
\usepackage{booktabs}
\usepackage{url}
\usepackage{todonotes}
\newcommand{\RR}{\mathbb{R}}
\usepackage{microtype}
\usepackage{multirow}
\usepackage{enumitem}
\usepackage{amsmath}
\usepackage{diagbox}
\usepackage{amsfonts}

\DeclareMathOperator{\concat}{concat}

\newcommand{\norm}[1]{\left\lVert #1 \right\rVert}
\newcommand{\blo}{\operatorname{blo}}
\newcommand{\bhi}{\operatorname{bhi}}

\renewcommand\footnotetextcopyrightpermission[1]{} % removes the ACM copyright footnote

\begin{document}

\title{Revisiting Eigenvalue-Only Divide-and-Conquer Tridiagonal Eigensolver from Algorithmic Perspective}

\author{Ruiyi Zhan}
\email{rui1@std.uestc.edu.cn}
\affiliation{%
  \institution{University of Electronic Science and Technology of China}
  \city{Chengdu}
  \country{China}
}

\author{Dajun Huang}
\email{djhuang_1@std.uestc.edu.cn}
\affiliation{%
  \institution{University of Electronic Science and Technology of China}
  \city{Chengdu}
  \country{China}
}

\author{Lu Shi}
\email{shilu75999@std.uestc.edu.cn}
\affiliation{%
  \institution{University of Electronic Science and Technology of China}
  \city{Chengdu}
  \country{China}
}

\author{Shaoshuai Zhang}
\authornote{Corresponding author.}
\email{szhang94@uestc.edu.cn}
\affiliation{%
  \institution{University of Electronic Science and Technology of China}
  \city{Chengdu}
  \country{China}
}

\renewcommand{\shortauthors}{Zhan et al.}

\begin{abstract}
Divide-and-conquer (D\&C) is the natural parallel strategy for symmetric tridiagonal eigenvalue decomposition, yet eigenvalue-only solvers fall back to serial QR algorithm because the quadratic transform state each merge carries exhausts memory at scale. In this paper, we demonstrate that each merge stage in D\&C needs only two boundary rows, the first and last rows of a node's eigenvector block, and these close into a recurrence whose parent rows follow from its children's and local merge data alone, so the state is carried up the tree with nothing replayed and persistent memory drops from quadratic to linear. Because this state keeps no replay workspace, adapting the schedule to each merge's post-deflation active rank is a scheduling choice, so both CPU and GPU spend finer-grained parallelism only where substantial secular work survives. We evaluate optimized CPU and GPU implementations, with complexity and correctness analysis, against QR, bisection, and D\&C routines in standard libraries.
\end{abstract}

\maketitle

\section{Introduction}

Symmetric eigenvalue decomposition (EVD) is a core primitive in numerical linear algebra. Given a real symmetric matrix $A \in \RR^{n \times n}$, EVD computes $A = Q\Lambda Q^T$, where $\Lambda$ stores the eigenvalues and $Q$ stores the corresponding orthonormal eigenvectors. Many applications use this spectral information in quantum chemistry~\cite{QuantumChemistry}, physics~\cite{grimes1987solution, probert2011electronic}, machine learning, and signal processing~\cite{xsvm, tensorsvm, shampoo, PCA}. In a substantial fraction of these workloads, however, the application needs only the eigenvalues~\cite{kpm2006,shklovskii1993}, or needs the eigenvalues before deciding which eigenvectors are necessary. This eigenvalue-only setting changes the optimization target, in that the solver should not pay quadratic memory cost for eigenvector data that will not be returned. The deeper consequence is parallel, since a quadratic working set forces even a parallel method onto a serial path at scale.

Dense symmetric EVD is normally reduced to a tridiagonal eigenproblem before the final spectral solve. LAPACK describes the standard pipeline as first reducing $A$ to a real symmetric tridiagonal matrix $T$, and then solving the tridiagonal problem for eigenvalues and, optionally, eigenvectors~\cite{LAPACK,MatrixComputation}. This separation makes the tridiagonal eigensolver a critical kernel in full EVD implementations. Beyond this pipeline, several applications pose a symmetric tridiagonal eigenproblem directly and often ask only for eigenvalues, such as Gauss quadrature, whose nodes are the eigenvalues of a symmetric tridiagonal Jacobi matrix~\cite{golubWelsch1969}, and Lanczos-based eigensolvers, whose Rayleigh--Ritz projection reduces a large sparse problem to a tridiagonal one whose eigenvalues approximate the wanted spectrum~\cite{BlockLanczos}. Once the problem reaches the tridiagonal matrix $T$, classical QR/QL routines are attractive for eigenvalue-only computation because they update only the diagonal and off-diagonal arrays and require little auxiliary storage~\cite{QRAlgorithm,LAPACK}. Their limitation is performance, since the computation is sequential in nature and exposes much less coarse-grained parallelism than D\&C.

Divide-and-conquer (D\&C) offers a more parallel alternative for symmetric tridiagonal eigenproblems. Cuppen's method recursively partitions $T$ into two smaller tridiagonal blocks coupled by a rank-one correction, solves each child, and reconciles them through a secular equation~\cite{Cuppen1981,gu1995divide,rank1EVD,secularEquation}. In classical analysis, deflation is a primary reason why D\&C can outperform serial alternatives on tridiagonal problems. When an eigenpair of a child block is already acceptable at the parent, its rank-one component is removed and the secular problem shrinks~\cite{gu1995divide,QRParlett}. How much merge work remains therefore depends on the eigenvalue distribution and on eigenvector structure, not only on problem size. The merge still exposes substantial parallelism across secular roots and, when eigenvectors are requested, assembles the parent eigenvectors from the child bases through dense matrix products. This structure is why production libraries ship D\&C routines and why D\&C can far outpace QR for large eigensystems~\cite{LAPACK,DivideAndConquer}, an advantage that widens on modern multicore CPUs and GPUs whose parallelism the sequential QR/QL recurrence cannot supply.

Yet D\&C is seldom the default for eigenvalue-only computation, and the obstacle is the transform state each merge must retain, not the parallelism D\&C already exposes.  Conventional D\&C must retain the dense eigenvector matrices of its subproblems to feed later merges, which costs $O(n^2)$ space.  The LAPACK \texttt{stedc} routine~\footnote{\href{https://netlib.org/lapack/explore-html/d3/d57/group__stedc_gaec55368cca7558e3ac13e04c1347bc27.html}{LAPACK \texttt{stedc} documentation}, lines 297--304} accordingly defaults the eigenvalue-only path to the QR algorithm (\texttt{steqr}) on account of this cost, 
% Values-only D\&C nonetheless retains the independent subproblems, secular roots, and coarse-grained merges that map naturally onto multicore CPUs and GPUs.
and what the libraries such as LAPACK avoid is only the quadratic transform state behind them.  
% If it stays quadratic, that state exhausts memory at large $n$ and forces even eigenvalue-only solves onto the sequential QR/QL path precisely where a parallel method matters most. 
% If the transform state is reduced to linear instead, the same parallel D\&C survives at those sizes and its schedule is no longer pinned to a dense per-node workspace, a routine no production GPU library exposes today.

% The systems entry point is direct. Asked for eigenvalues alone, the internal values-only D\&C path still queries quadratic scratch, roughly 6~GiB already at $n = 16{,}384$ and out of memory by $n = 65{,}536$, and spends much of each merge replaying the same per-level transforms back to the leaves to re-derive the two rows each merge reads. That an $O(n)$ result demands an $O(n^2)$ working set and repeated leaf-ward replay is what turned our attention to what this state is for, and whether a running state carried up the tree could replace it.

% The mismatch is that the vector each rank-one merge needs is itself compact, yet what a library retains to form it is not. Expressed in the two child eigenbases, that vector comprises only the last row of the left child's eigenvector matrix and the first row of the right child's. Calling these two rows a node's boundary rows, each merge reads just the boundary rows flanking the split, one per child, while standard D\&C in LAPACK, MKL and cuSOLVER still retains enough child eigenvector data to regenerate them at later levels.

This paper asks whether that D\&C path actually needs to retain all replayable transform states. Lazy replay comes closest to removing that state, since it already recognizes that a merge consumes only boundary rows and replaces the eager GEMMs with a chained GEMV that recovers each row on demand~\cite{chainedDC}. However, lazy replay still recovers those rows by replaying the stored per-level transforms back to the leaves, so the accumulated transforms remain and the footprint stays quadratic. We show instead that the boundary rows close into a recurrence, in which each parent's boundary rows follow from its children's boundary rows and local merge data alone.  We turn this recurrence into a boundary-row D\&C algorithm that propagates the state up the tree. Instead of replaying the stored transforms, the algorithm updates the boundary rows needed by later merges directly. Closing the recurrence is not free, since materializing each local transform would restore the quadratic state. We therefore consume that transform in a streaming reconstruction, which keeps the persistent state linear while preserving the same secular-equation structure as conventional D\&C.

% Reducing the persistent state to linear does more than restore feasibility. Because the boundary-row state carries no dense per-node replay workspace, the parallel schedule is free to adapt to how much merge work each level actually keeps after deflation rather than being pinned to problem size, since what had pinned it was the workspace footprint, not its arithmetic. The same representation that shrinks memory therefore also lets the merge parallelism follow matrix structure, which is what the algorithmic perspective buys beyond feasibility (Section~\ref{sec:experiments}).

% The rest of the paper follows this data dependency from library behavior to algorithm, implementation, and measurement. Section~\ref{sec:background} reviews how existing tridiagonal solvers expose the QR-versus-D\&C trade-off, Section~\ref{sec:method} formalizes the boundary-row state that replaces dense replay, Section~\ref{sec:implementation} realizes the same state contract on CPU and GPU, Section~\ref{sec:experiments} measures the resulting tridiagonal-stage prototypes against QR/QL, MRRR's eigenvalue-only \texttt{xSTEGR} wrapper, bisection, internal values-only D\&C, and cuSOLVER D\&C baselines, and Section~\ref{sec:conclusion} concludes.

This paper makes the following contributions:
\begin{itemize}
    \item We formalize this boundary-row recurrence as a closed state invariant, so eigenvalue-only D\&C needs neither the dense child basis nor the replayable transform state that lazy replay must retain.
    \item We turn this invariant into a boundary-row D\&C algorithm, realized by a streaming secular reconstruction and a compact-$\delta$ representation, whose linear auxiliary memory lets the parallel D\&C merge run at the large-$n$ scales where the quadratic replay state would otherwise force eigenvalue-only solves back onto the serial QR/QL path.
    \item We prove that the propagated boundary rows reproduce the secular-equation inputs of conventional D\&C exactly, and characterize the resulting time and space complexity.
    \item We implement CPU and GPU tridiagonal-stage prototypes whose parallel scheduling is hardware-aware and adapts to matrix structure, and evaluate them against QR/QL, MRRR's eigenvalue-only \texttt{xSTEGR} wrapper, bisection, internal values-only D\&C, and cuSOLVER D\&C on varied problem sizes and spectral structures.
\end{itemize}

\section{Background and Related Work}
\label{sec:background}

\subsection{Symmetric EVD and Tridiagonal Solvers}

Dense symmetric EVD has long been organized around a two-stage pipeline whose shape reflects where the cost lies~\cite{MatrixComputation,LAPACK}. A real symmetric $A$ is first reduced by orthogonal transformations to a real symmetric tridiagonal $T$ that shares its eigenvalues, a direct computation that terminates in a fixed number of steps~\cite{MatrixComputation,DBBR}. The spectral work then falls entirely on $T$, whose narrow bandwidth is what makes the eigenproblem tractable, and it is at this tridiagonal stage that decades of algorithmic effort have concentrated. These solvers coexist as alternatives, where QR iteration, bisection, D\&C, and MRRR all attack the same tridiagonal $T$ and differ mainly in how much spectral information they return and at what cost in work and memory~\cite{QRAlgorithm,Bisection,gu1995divide,MRRR}. 
% This paper stays within that stage, and within the case that asks only for eigenvalues, which the libraries below treat as a distinct path rather than a special case of the eigenvector computation.

LAPACK exposes these alternatives as interchangeable tridiagonal eigensolvers that trade accuracy, workspace, and parallelism differently. QR/QL routines such as \texttt{xSTERF} compute all eigenvalues without eigenvectors through a square-root-free variant~\cite{QRSymHandbook,QRSqrtFree,QRValuesOnly,QRParlett,LAPACK}. Bisection locates eigenvalues by intervals~\cite{Bisection,LAPACK}, while MRRR uses relatively robust representations and shifted factorizations, with parallel variants also available~\cite{MRRR,MRRR2,PMRRR,ScaLAPACK}; in the eigenvalue-only interface, \texttt{xSTEGR}/\texttt{xSTEMR} is called with \texttt{JOBZ=N}. D\&C, in routines such as \texttt{xSTEDC}, targets all eigenvalues and eigenvectors and can far outrun QR for large problems, but needs substantially more workspace~\cite{LAPACK,gu1995divide}.

\subsection{Divide-and-Conquer for Tridiagonal EVD}

D\&C solves the tridiagonal eigenproblem by recursively reducing it to diagonal-plus-rank-one problems. Given an irreducible tridiagonal matrix $T$, the algorithm splits it at an off-diagonal entry and writes
\begin{equation}
    T =
    \begin{bmatrix}
        T_1 & 0 \\
        0 & T_2
    \end{bmatrix}
    + \rho u u^T,
    \label{eq:dc-split}
\end{equation}
where $T_1$ and $T_2$ are smaller tridiagonal matrices, $u$ has nonzeros only at the split boundary, and the scalar $\rho$ is the rank-one coupling strength fixed by the off-diagonal entry removed at the split~\cite{Cuppen1981,gu1995divide}. If the child decompositions are
\begin{equation}
    T_i = Q_i \Lambda_i Q_i^T,\quad i \in \{1,2\},
\end{equation}
then congruence by the block-diagonal child eigenbasis leaves the same rank-one coupling scalar $\rho$ and turns $u$ into $v=(Q_1\oplus Q_2)^Tu$, so the parent merge becomes
\begin{equation}
    D + \rho v v^T,
    \quad
    D = \Lambda_1 \oplus \Lambda_2,
    \quad
    v =
    \begin{bmatrix}
        \pm Q_1^T e_m \\
        Q_2^T e_1
    \end{bmatrix},
    \label{eq:merge-problem}
\end{equation}
where $\Lambda_1\oplus\Lambda_2$ is $\begin{bmatrix}
        \Lambda_1 & 0 \\
        0 & \Lambda_2
    \end{bmatrix}$, $e_1=[1, 0,...,0]^T$ and $e_m=[0,...,0,1]^T$. Thus, the merge depends on the last row of $Q_1$ and the first row of $Q_2$. The eigenvalues of the parent are then the roots of a secular equation associated with $D+\rho vv^T$~\cite{rank1EVD,secularEquation}. This recursion exposes parallelism at several granularities, since independent subproblems at the same level, the secular roots within a merge, and the eigenvector-update GEMMs can all proceed concurrently, which is what makes D\&C attractive on multicore CPUs and GPUs~\cite{DivideAndConquer}.
Deflation further prunes the secular work at each merge: eigenpairs with negligible rank-one components or near-coincident eigenvalues pass through unchanged~\cite{gu1995divide}, leaving an active rank that is data-dependent and varies across the tree.

After the secular equation is solved, conventional D\&C forms the parent eigenvector matrix as a dense product of the child bases with the rank-one eigenvectors. This step is both D\&C's speed advantage and its memory cost, since it keeps dense eigenvector matrices at every intermediate level, which is why LAPACK pairs D\&C's large-matrix speedup over QR with quadratic workspace~\cite{LAPACK}. The trade-off suits eigenvector output but is wasteful when only eigenvalues are wanted. Optimizing this eigenvector step has been a continuing effort, from Gu's fast-multipole acceleration of the secular and eigenvector work~\cite{guAcceleratedSVD} to faster rank-one merges~\cite{LiangEigenUpdate}, but these make the work cheaper rather than unnecessary. Even for eigenvalue-only output, each merge still needs its children's boundary rows, so low-memory paths give up the parallel merge while D\&C paths retain quadratic replay or transform state.

\section{Method}
\label{sec:method}

We approach eigenvalue-only D\&C as a question about the dependency structure of the merge tree rather than about eigenvector arithmetic. Since only eigenvalues leave the root, every merge below it needs just two boundary rows, so the one design choice is how those two rows reach each merge. Lazy replay answers this from a numerical linear algebra standpoint, where each boundary row is exactly the chain of per-level transforms that carries a leaf row up the tree, so it need not be stored and is replayed from the leaves whenever a later merge asks for it. We read the same quantity from an algorithmic standpoint, where the boundary row is not an expression to unroll back to the leaves but the output of a state transition that forms each node's two rows from its children's two rows and its local transform alone. That transition turns the chain into a recurrence carried up the tree with nothing ever replayed, which reduces the persistent state from quadratic to linear (Section~\ref{sec:complexity}) while constructing the same secular problems as conventional D\&C (Theorem~\ref{thm:br-correct}).

\subsection{Lazy Replay: Numerical Linear Algebra Perspective}
% The background reduces the eigenvalue-only D\&C problem to a precise data-dependency question: which part of the child eigenvector basis is required by later secular merges?
As illustrated in Figure~\ref{fig:dc-common}, the original tridiagonal problem is split into subproblems whose eigenvectors are combined level by level. Full-eigenvector D\&C forms each parent eigenvector matrix as a dense product $(Q_L\oplus Q_R)S$ of the block-diagonal child basis, built from the child eigenvector matrices $Q_L$ and $Q_R$, with the secular eigenvectors $S$. This dense update at every internal node makes D\&C fast but also makes its state quadratic. When only eigenvalues are requested, the merge needs far less, since forming the secular vector $z$ requires only two boundary rows, the last row of the left child block and the first row of the right child block, rather than the full child basis.

\begin{figure}[t]
    \centering
    \includegraphics[width=0.84\columnwidth]{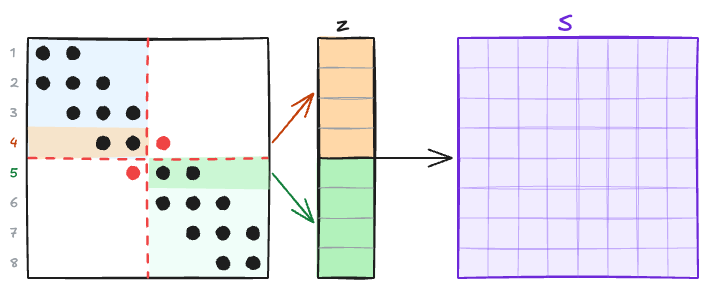}
    \Description{A divide-and-conquer merge diagram showing that the secular vector is formed from the last row of the left child block and the first row of the right child block.}
    \caption{A single D\&C merge. The secular vector $z$ is formed from only two boundary rows, the last row of the left child (orange) and the first row of the right child (green). The secular solver then turns $z$ into the local transform $S$.}
    \label{fig:dc-common}
\end{figure}

Lazy replay exploits this fact to avoid eagerly forming and storing the full intermediate eigenvector matrices~\cite{chainedDC}. As drawn in Figure~\ref{fig:lazy-replay}, it recovers each boundary row $b$ by a chained GEMV that carries a leaf boundary row $b_{leaf}$ up through the stored transforms,
\begin{equation}
    b = b_{\mathrm{leaf}}\,S_1 S_2 \cdots S_d,
    \label{eq:chain-gemv}
\end{equation}
as implemented by LAPACK's \texttt{DLAEDA}. Lazy replay thus removes the eager dense updates of full-eigenvector D\&C and brings the values-only work down to an \(O(n^2)\) merge path, but the transforms it must retain for replay still form \(O(n^2)\) state. We view lazy replay from a numerical linear algebra standpoint, since it is natural to substitute the chained GEMMs in D\&C that involve eigenvectors with chained GEMVs when only single rows are required to construct the $z$ vector illustrated in Figure~\ref{fig:dc-common}.

Replay always reaches back to the leaves, yet it need not. Once a node's boundary rows are known, the parent's follow directly from the local transform without revisiting the leaves. The dependency is therefore closed under boundary-row propagation, in that the boundary rows of a parent can be computed from the boundary rows and local merge information of its children alone.

\subsection{Boundary-Row State: Algorithmic Perspective}
\label{sec:boundary-row-state}

Because eigenvalue-only D\&C is a divide-and-conquer algorithm, the question is whether each parent can be served by compact child state alone. Let $\blo(\cdot)$ and $\bhi(\cdot)$ return the first and last row of a matrix. A single merge (Figure~\ref{fig:dc-common}) needs only $z=[\bhi(Q_L)~|~\blo(Q_R)]^T$, formed from the last row of the left child $Q_L$ and the first row of the right child $Q_R$, so we attach to every node a boundary-row state $B=(B_{lo},B_{hi})$ with $B_{lo}=\blo(Q)$ and $B_{hi}=\bhi(Q)$, and ask whether this state closes into a recurrence up the tree.

This state drives one transition per level, drawn flowing up the tree in Figure~\ref{fig:br-propagation}. At level $i$ the two child boundary rows form the parent secular vector
\begin{equation}
    z_i=[B(i-1)_{hi}^L~|~ B(i-1)_{lo}^R]^T,
    \label{eq:zi}
\end{equation}
whose secular solve returns the parent eigenvalues $\Lambda$ together with the local transform $S_i$ of $D+\rho z_iz_i^T$. Passing the two selected child rows through $S_i$ produces the parent's own boundary rows,
\begin{equation}
    B(i)_{lo}=[B(i-1)_{lo}^L~|~0]\,S_i,
    \label{equa:left}
\end{equation}
\begin{equation}
    B(i)_{hi}=[0~|~B(i-1)_{hi}^R]\,S_i,
    \label{equa:right}
\end{equation}
which are exactly the rows the next merge reads. Since these parent rows re-enter the level-$(i+1)$ secular vector in the same form as Eq.~\eqref{eq:zi}, the map from child boundary state to parent boundary state is the same at every level, so the boundary rows are not merely read off at each merge but form a closed state that a single transition rule carries from the leaves up to the root.

Deflation is what makes this state cheap rather than merely correct, and it acts on the boundary rows (BR) directly. A merge deflates a root when its rank-one component is negligible or when two poles nearly coincide~\cite{gu1995divide}, and the first test reads exactly the entries BR already holds, since every component of $z_i$ in Eq.~\eqref{eq:zi} is a boundary-row entry, so the negligible-component decision inspects the two boundary rows rather than any eigenvector block, while a close-pole rotation touches only the two corresponding boundary-row entries. BR therefore makes the same deflation decisions as conventional D\&C from strictly less state, and it then propagates only the surviving active rows, so the boundary-row state handed to the parent is already the post-deflation active set. The number of survivors is the active rank $K$, and the ratio $r=K/m$ over a merge's $m$ roots is the single quantity the per-merge cost of Section~\ref{sec:complexity} and the schedules of Section~\ref{sec:implementation} both follow. Algorithm~\ref{alg:boundary-row-dc} states the resulting boundary-row merge in full.

\begin{algorithm}[t]
\caption{Boundary-row D\&C for eigenvalue-only tridiagonal EVD}
\label{alg:boundary-row-dc}
\begin{algorithmic}[1]
\REQUIRE Symmetric tridiagonal subproblem $T_v$
\ENSURE Eigenvalues $\Lambda_v$ and boundary state $B_v=(B_{lo},B_{hi})$
\IF{$T_v$ is a leaf problem}
    \STATE Solve the leaf by a small tridiagonal eigensolver.
    \RETURN $\Lambda_v$ and its boundary rows $B_v=(\blo(Q_v),\bhi(Q_v))$
\ENDIF
\STATE Split $T_v$ into left and right child problems.
\STATE Recursively compute the child spectra and boundary states $(\Lambda_L,B_L)$ and $(\Lambda_R,B_R)$.
\STATE Form $D=\Lambda_L\oplus\Lambda_R$ and $z=\concat(B_{hi}^L,B_{lo}^R)$.
\STATE Deflate roots with negligible $z$ entries or near-coincident poles, acting on the boundary rows and their metadata, leaving $K$ active roots.
\STATE Solve the $K$ active secular equations for $\Lambda_v$ and the local transform $S_v$.
\STATE Update the parent rows by the transition of Eqs.~\eqref{equa:left} and~\eqref{equa:right}, $B_{lo}=[B_{lo}^L~|~0]\,S_v$ and $B_{hi}=[0~|~B_{hi}^R]\,S_v$, consumed as the streaming reconstruction of Section~\ref{sec:streaming}.
\RETURN $\Lambda_v$ with $B_v=(B_{lo},B_{hi})$, or $\Lambda_v$ alone at the root
\end{algorithmic}
\end{algorithm}

\begin{figure*}[t]
    \centering
    \begin{minipage}[t]{0.48\textwidth}
        \centering
        \includegraphics[width=0.72\linewidth]{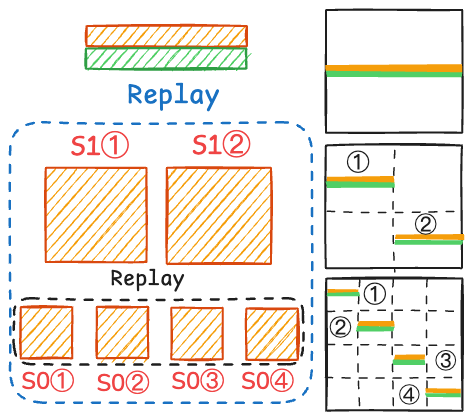}
        \subcaption{Lazy replay in values-only D\&C following Rutter~\cite{chainedDC}. The left panel shows the intermediate secular eigenvector matrices $S_i$ stored at internal nodes of the D\&C tree, from the leaf factors $S0$ up through $S1$ and deeper. Because these matrices are local merge factors, recovering a boundary row needs a chained GEMV that replays them back to the leaves. The right panels show the same merges from the matrix side reused in panel (b), where each circled index marks one replayed row.}
        \label{fig:lazy-replay}
    \end{minipage}\hfill
    \begin{minipage}[t]{0.48\textwidth}
        \centering
        \includegraphics[width=0.88\linewidth]{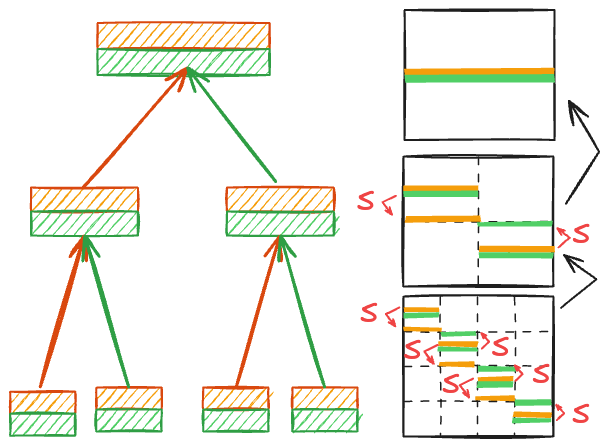}
        \subcaption{Boundary-row propagation in BR. Each node keeps only its two boundary rows (orange and green) rather than a full eigenvector block, and at every merge the local transform $S$ maps the children boundary rows to the parent boundary rows through the recurrence of Eqs.~\eqref{equa:left} and \eqref{equa:right}, so the rows flow up the tree (left) until only two rows remain at the root. A node supplies $B_{hi}$ when it is a left child and $B_{lo}$ when it is a right child, its sibling supplying the other half.}
        \label{fig:br-propagation}
    \end{minipage}
    \caption{Values-only D\&C state carried across levels. Rutter-style lazy replay (a) keeps intermediate secular eigenvector matrices and replays their chain to recover boundary rows, whereas BR (b) propagates only the two boundary rows of each node.}
    \Description{Side-by-side comparison of lazy replay, which stores intermediate eigenvector matrices and replays them, and boundary-row propagation, which keeps only two rows per node.}
    \label{fig:replay-vs-br}
\end{figure*}

At the root, only eigenvalues are returned. Boundary rows are nevertheless propagated through internal nodes because they are exactly what ancestor secular equations read. In implementation, the state a node returns is just its two boundary rows, and the deflation permutations, column signs, and sorting that each merge applies act on those rows through metadata rather than on any materialized matrix. BR forms these rows from the same local secular eigenvectors as conventional D\&C, built from the corrected L\"owner update vector \(\widehat w\) of Gu and Eisenstat~\cite{gu1995divide} rather than the raw \(z\), and only streams the two selected rows instead of materializing the transform.

\subsubsection{Algebraic Correctness}

\begin{lemma}[Boundary-state sufficiency]
\label{lem:boundary-sufficiency}
For any D\&C merge with child eigenvector matrices $Q_L$ and $Q_R$, the secular vector $z$ in Eq.~\eqref{eq:merge-problem} is determined by $\bhi(Q_L)$ and $\blo(Q_R)$.
\end{lemma}

\begin{proof}
The rank-one coupling in Eq.~\eqref{eq:dc-split} acts only on the last coordinate of the left child and the first coordinate of the right child. After transforming to the child eigenvector bases, these two coordinate vectors become $Q_L^Te_{n_L}$ and $Q_R^Te_1$. Their entries are precisely the last row of $Q_L$ and the first row of $Q_R$, giving Eq.~\eqref{eq:merge-problem}.
\end{proof}

\begin{lemma}[Selected-row multiplication]
\label{lem:selected-row}
Let $Q\in\mathbb{R}^{m\times k}$, $S\in\mathbb{R}^{k\times \ell}$, and let $\sigma$ be any row list, possibly unordered and with repetitions. Then $(QS)_\sigma = Q_\sigma S$, since each output row of $QS$ is formed from the corresponding row of $Q$ alone.
\end{lemma}

\begin{theorem}[BR computes the same eigenvalues as conventional D\&C in exact arithmetic]
\label{thm:br-correct}
Assume BR and conventional D\&C use the same split tree, deflation decisions, column ordering, column signs, and secular-equation convention. In exact arithmetic, every internal BR merge constructs the same secular problem as conventional D\&C, and the root eigenvalues returned by BR are identical to those returned by conventional D\&C.
\end{theorem}

\begin{proof}
The proof is by induction over the merge tree. At leaves, BR solves the same leaf eigenproblem as conventional D\&C and returns the two boundary rows of the same local eigenvector matrix. For an internal node, the induction hypothesis gives the same child spectra and the same child boundary rows. By Lemma~\ref{lem:boundary-sufficiency}, BR constructs the same $z$ and therefore the same diagonal-plus-rank-one secular problem as conventional D\&C. The local eigenvalues are therefore the same. By Lemma~\ref{lem:selected-row}, the parent's two boundary rows can be computed from the selected child rows alone, without forming the full child basis. Hence the induction hypothesis also holds for the parent boundary state. Applying the argument up to the root proves the claim.
\end{proof}

\subsubsection{Conditioned Error Propagation}
\label{sec:conditioned-error}

BR removes storage and arithmetic associated with unrequested rows, but it still relies on the same local numerical ingredients as D\&C, including secular equation solves, deflation handling, and products used to form selected rows. Our stability statement is therefore conditioned on standard local stability interfaces rather than a fresh proof of the full LAPACK-style secular solver.

Let $\widehat B_v$ be the computed boundary state at node $v$. If the two child boundary-row errors satisfy
\begin{equation}
    \norm{\widehat{\bhi}_L-\bhi_L}_2\le \eta_L,\qquad
    \norm{\widehat{\blo}_R-\blo_R}_2\le \eta_R,
\end{equation}
then the error in the constructed secular vector satisfies
\begin{equation}
    \norm{\widehat z-z}_2^2
    =
    \norm{\widehat{\bhi}_L-\bhi_L}_2^2+
    \norm{\widehat{\blo}_R-\blo_R}_2^2
    \le
    \eta_L^2+\eta_R^2.
    \label{eq:z-error}
\end{equation}

Thus BR does not amplify the child boundary error when forming $z$ beyond the norm of the concatenation. The remaining local error is the error of the secular solve and the selected-row update.

Carrying \eqref{eq:z-error} through to the root, the BR-computed eigenvalues satisfy $\norm{\widehat\lambda-\lambda}_2\le\eta_{\mathrm{solver}}+\kappa_{\mathrm{pert}}\,\eta_{\mathrm{prob}}$, where $\eta_{\mathrm{prob}}$ bounds how far the BR-constructed root secular problem lies from the exact one, $\kappa_{\mathrm{pert}}$ is the conditioning of the eigenvalue map under that perturbation, and $\eta_{\mathrm{solver}}$ is the error of the secular root solver itself. The solver and root-conditioning terms are the same ones a conventional lazy-replay merge incurs, so BR stays at the same eigenvalue-error order. A full conditioned proof is deferred to the supplemental material.

Table~\ref{tab:accuracy-br} reports high-precision spot checks on the CPU
input families.  We compute selected eigenvalues with 256-bit MPFR Sturm
bisection and report
\[
e_{\mathrm{MPFR}} =
    \frac{\|\widehat\lambda-\lambda_{\mathrm{MPFR}}\|_\infty}
         {\max(1,\|\lambda_{\mathrm{MPFR}}\|_\infty)}
\]
at the largest-discrepancy indices found by the common-reference sweep.  The
result exposes the important accuracy point. On the nontrivial rows the D\&C
paths are closer to the high-precision reference than QR/QL, and BR stays at the
same scale as the internal values-only D\&C path.  All BR runs use the
compact-\(\delta\) reconstruction of Section~\ref{sec:streaming}, which keeps
the pole-adjacent denominators of each secular root that direct subtraction
would lose.

\begin{table}[t]
\caption{Normalized forward difference \(e_{\mathrm{MPFR}}\) to a 256-bit MPFR Sturm bisection reference at selected high-discrepancy indices. The inputs are tridiagonal matrices of size $16{,}384 \times 16{,}384$, each featuring a distinct eigenvalue spectrum.}
\label{tab:accuracy-br}
\centering
\footnotesize
\setlength{\tabcolsep}{3.0pt}
\begin{tabular}{@{}cccc@{}}
\toprule
\diagbox[]{Matrix Type}{Routine}  & QR/QL & Internal D\&C & BR \\
\midrule
Uniform  & \(2.99\times10^{-14}\) & \(6.07\times10^{-16}\) & \(3.91\times10^{-16}\) \\
% Uniform & \(49{,}152\) & \(5.29\times10^{-14}\) & \(1.91\times10^{-16}\) & \(7.46\times10^{-16}\) \\
Normal  & \(4.65\times10^{-14}\) & \(5.94\times10^{-16}\) & \(2.92\times10^{-16}\) \\
% Normal & \(49{,}152\) & \(9.85\times10^{-14}\) & \(1.11\times10^{-15}\) & \(7.25\times10^{-16}\) \\
Toeplitz  & \(8.99\times10^{-15}\) & \(7.68\times10^{-16}\) & \(7.77\times10^{-16}\) \\
% Toeplitz & \(49{,}152\) & \(1.64\times10^{-14}\) & \(1.02\times10^{-15}\) & \(1.21\times10^{-15}\) \\
Clustered & \(2.99\times10^{-17}\) & \(8.11\times10^{-17}\) & \(8.11\times10^{-17}\) \\
% Clustered & \(49{,}152\) & \(6.33\times10^{-17}\) & \(1.74\times10^{-16}\) & \(1.74\times10^{-16}\) \\
Glued Wilkinson & \(7.40\times10^{-15}\) & \(4.31\times10^{-16}\) & \(6.10\times10^{-16}\) \\
% Glued Wilkinson & \(49{,}152\) & \(4.7\times10^{-14}\) & \(1.2\times10^{-16}\) & \(1.2\times10^{-16}\) \\
\bottomrule
\end{tabular}
\end{table}

\subsection{Complexity}
\label{sec:complexity}

% The relevant comparison is between values-only D\&C paths, not between BR and the public QR/QL routine \texttt{xSTERF}. The public \texttt{xSTEDC} values-only interface may choose the low-workspace QR/QL path, whereas LAPACK's internal values-only D\&C machinery, represented by \texttt{DLAED0(ICOMPQ=0)}, retains enough replay state for \texttt{DLAEDA} to reconstruct future merge vectors. BR targets this second point in the design space: it keeps the D\&C merge and secular-solve structure, but changes the vector-derived state that is carried across levels.

BR replaces lazy replay's intermediate matrix-state with boundary-row state, giving two results of different character. In space it reduces the values-only D\&C footprint from \(O(n^2)\) to \(O(n)\), a clean spectrum-independent gain. In time no such clean bound exists, since the cost is governed by deflation rather than \(n\); what holds for every input is that BR does strictly less merge-path work than lazy replay.

In space, BR's auxiliary state is a small constant number of length-\(n\) arrays, \(O(n)\) excluding the input tridiagonal arrays and output eigenvalues. At any time it holds the running pole array, the two boundary rows \(B_{lo}\) and \(B_{hi}\), the compact secular data of the merges in flight (one origin index, one offset \(\tau\), and two pole-adjacent \(\delta\) entries per active root), and \(O(n)\) integer permutation and deflation metadata. The one \(O(n^2)\) object, the orthogonal transform \(S\) of a secular solve, is never materialized but consumed column by column in the streaming reconstruction of Section~\ref{sec:streaming}. By contrast, lazy replay stores dense secular-vector blocks and replay metadata, giving \(O(n^2)\) real workspace and \(O(n\log n)\) integer metadata. The supplemental material derives the \(O(n)\) bound, including the level-batched case where the active ranks partition at most \(n\).

For time, both paths run the same secular solves and deflation and differ only in how \(z\) reaches each merge. Lazy replay reconstructs it by replaying the stored transforms back to the leaves~\cite{chainedDC} (LAPACK's \texttt{DLAEDA}), whereas BR applies \(S\) to its boundary rows once per non-root merge, so BR's work is lazy replay's with the replay term removed and every other term shared, strictly less on any input. The absolute cost, however, is not a function of \(n\). With post-deflation active rank \(K_j=r_j m_j\) at merge \(j\), propagation and the secular solve are each \(\Theta(K_j^2)\) (the latter up to a data-dependent iteration count near clustered poles), so BR's time is \(\Theta(\sum_j (r_j m_j)^2)\).  This hits the \(O(n^2)\) ceiling only when nothing deflates and is otherwise set by the active-rank profile \(r=K/m\) reported in the supplemental material, which is why equal-\(n\) inputs span an order of magnitude in Table~\ref{tab:cpu-dsterf-vs-br}.  On a balanced tree that shared ceiling evaluates to \(\tfrac{4}{3}n^2\) for lazy replay against \(n^2\) for BR. Deflation compounds this advantage, because BR propagates only the surviving rows, each deflated root shrinks its boundary state, streaming reconstruction, and schedule at once, whereas lazy replay's \(O(n^2)\) footprint stays in place and the same root shrinks only its local secular solve.

Table~\ref{tab:complexity} summarizes the time and space complexity of the values-only paths and gives representative elapsed time at \(N=65{,}536\) on the spectrum-prescribed uniform input under the 56-thread CPU configuration used in Table~\ref{tab:cpu-dsterf-vs-br}.  The CPU baseline set includes QR/QL, bisection, the MRRR eigenvalue-only \texttt{xSTEGR(JOBZ=N)} wrapper, parallel bisection, and lazy-replay D\&C.  For \texttt{xSTEGR}, the linear auxiliary-workspace qualifier applies only to the \texttt{JOBZ=N} call, not to the full eigenvector-producing MRRR path.  In addition to the shared-memory LAPACK/MKL bisection routine, we include \texttt{PDSTEBZ}, the ScaLAPACK distributed bisection routine~\cite{ScaLAPACK}, on 56 MPI ranks, to separate the algorithmic cost of full-spectrum bisection from the limited thread scaling of the sequential interface.

\begin{table}[t]
\centering
\caption{Complexity summary with representative elapsed time on the uniform \(N=65{,}536\) case.}
\label{tab:complexity}
\footnotesize
\setlength{\tabcolsep}{2.5pt}
\begin{tabular}{@{}lccr@{}}
\toprule
Method & Time & Space & Time at \(65{,}536\) \\
\midrule
QR/QL (\texttt{xSTERF}) & $O(n^2)$ & $O(n)$ & 56.35~s \\
Bisection (\texttt{xSTEBZ}) & $O(n^2)$ & $O(n)$ & $>10$~min \\
Par.\ bisection (\texttt{PDSTEBZ}) & $O(n^2)$ & $O(n)$ & 18.34~s \\
MRRR (\texttt{xSTEGR}) & $O(n^2)$ & $O(n)$ & 88.85~s \\
Lazy-replay D\&C & $O(n^2)$ & $O(n^2)$, $O(n\log n)$~int & OOM (>96~GiB) \\
BR (this work) & $O(n^2)$ & $O(n)$ & 168~ms \\
\bottomrule
\end{tabular}
\end{table}

\section{Implementation}
\label{sec:implementation}

The GPU and CPU paths realize the same boundary-row modification to tridiagonal D\&C, sharing the split tree, deflation, sorting, secular solves, and denominator reconstruction while replacing the persistent full-eigenvector or replay state with the boundary-row state of Section~\ref{sec:boundary-row-state}. They differ only in platform, where the GPU maps the shared state contract to resident CUDA kernels and the CPU exposes it through a LAPACK-style prototype for controlled comparison with existing routines. The implementation has two common requirements: reconstruct the selected rows without materializing the local transform, and adapt the available parallelism to the secular work that remains after deflation. This section first describes the streaming reconstruction, then the deflation-aware parallel schedule, and finally the CPU and GPU realizations.

\subsection{Streaming Secular Reconstruction}
\label{sec:streaming}
After each secular solve produces the eigenvalues, the parent boundary rows follow by passing the secular eigenvectors through the two selected child rows. Assembling the whole secular eigenvector block first and then applying the two matrix-vector products would reintroduce the quadratic storage which is removed by the boundary-row state, so we consume the same matrix-vector product as a stream of columns instead.  Since each root yields one column on its own, a small working buffer folds each column into the two parent boundary rows and discards it, turning one block product into many short vector updates with the same total arithmetic and no extra main-memory traffic.  Streaming rather than batching this product is also what turns the memory saving into a scheduling one, since a materialized block would fix a per-merge workspace whose footprint, not its arithmetic, bounds how merges can be co-scheduled, whereas the column stream lets the machine place work by the active rank each merge keeps after deflation.

Each column's entries still need denominators, but evaluating them by direct subtraction loses accuracy near a pole. As shown in Figure~\ref{fig:near-pole}, we therefore cache the few cancellation-sensitive near-pole entries per root, computed accurately once, while cheap far-pole entries are evaluated directly wherever needed, keeping the per-merge state linear without sacrificing accuracy, with the compact-\(\delta\) construction detailed in the supplemental material. Both hardware paths build on this streaming update with near-pole caching.

\begin{figure}[t]
    \centering
    \includegraphics[width=0.94\columnwidth]{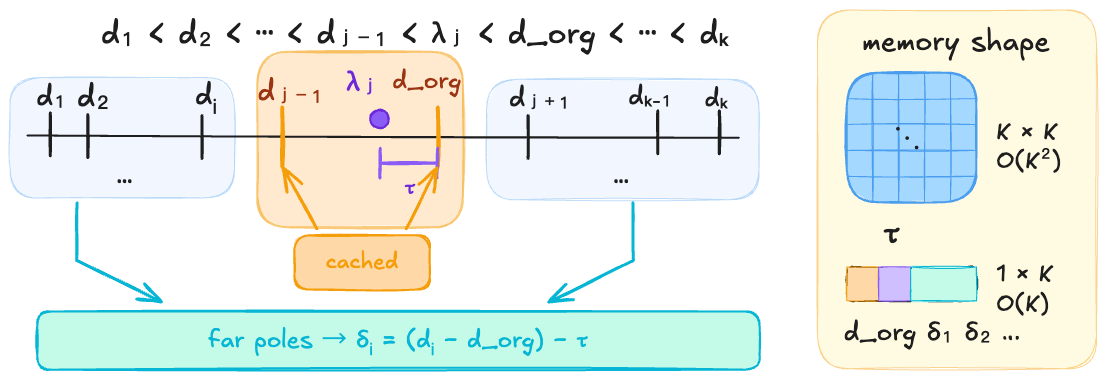}
    \Description{Diagram of near-pole secular denominator handling, with cached entries around each root and recomputed far-pole entries.}
    \caption{Near-pole denominator handling. The entries around a root $\lambda_j$ and its origin pole $d_{\mathrm{org}}$ are cached (orange), while the far poles are recomputed through $\delta_i=(d_i-d_{\mathrm{org}})-\tau$ (teal), shrinking the stored state from the dense $O(K^2)$ denominator table to a compact $O(K)$ record.}
    \label{fig:near-pole}
\end{figure}

\subsection{Deflation-Aware Parallel Scheduling}
\label{sec:aware-scheduling}

How much parallel work a merge carries is set by deflation, not by its size. For a merge of size \(m\), let \(K\) be the roots that survive deflation and \(r=K/m\) the active-rank ratio. Because \(r\) swings widely across spectra and across levels, a schedule read from \(m\) alone can crowd workers onto a merge whose secular work has already collapsed. BR reads this signal for free, since it learns \(K\) as it deflates the boundary rows, so one quantity prunes the persistent state, shortens the streaming reconstruction, and sizes the scheduled work.

Guided by \(r\), BR chooses between two complementary axes of parallelism. Broad lower levels run their many independent merges concurrently, while a narrow upper level spreads its single high-\(r\) merge's active roots across the same workers, and a low-\(r\) merge stays lightweight without an inner team. The switch is cheap because every worker acts on the same two boundary rows, so it changes only the schedule and never the persistent state, whereas a replay path would first rebuild those rows along its transform chain. Figure~\ref{fig:scheduler-ablation} measures this choice against fixed merge-level and secular-root schedules on the CPU, and the CPU and GPU paths below realize the same policy with platform-specific thread teams and kernel granularities.

\subsection{Hardware-Specific Implementation}
\subsubsection{CPU Implementation}

To keep a controlled baseline for direct comparison, the CPU prototype realizes the BR algorithm as a private values-only path that runs alongside, but never replaces, the public eigensolver. It keeps the usual outer structure of a tridiagonal solver, validating arguments, splitting at negligible off-diagonal entries, and scaling unreduced blocks, before routing small blocks to the standard root-free QR iteration and larger blocks into the BR D\&C.

The driver keeps the standard split tree but stores only the first and last row of each active block, produced by the usual dense solver at leaves. At each non-root merge the rank-one update vector $z=[\bhi(Q_L)~|~\blo(Q_R)]$ is read directly from the child boundary rows, and sorting, deflation, permutation, and close-eigenvalue rotations act on these two rows rather than on full eigenvector blocks.  Because every entry of $z$ is a boundary-row entry, the negligible-component deflation test reads the two rows BR already holds, and only the surviving active rows are streamed to the parent, so deflation both decides on and prunes the same compact state and fixes the active rank $r=K/m$ that the schedule below reads. The parent rows then follow from the streaming update of Section~\ref{sec:streaming}. Each active root $\lambda_j$ contributes one secular column $y_j$, reduced against the child boundary rows entering Eqs.~\eqref{equa:left} and \eqref{equa:right} at a cost of two dot products per entry, with denominators reused from the compact form of Section~\ref{sec:streaming},
\[
    \lambda_j = d_{\mathrm{origin}}+\tau_j, \qquad \delta_i = d_i - d_{\mathrm{origin}} - \tau_j,
\]
so no dense per-merge secular-vector storage is introduced.

The CPU scheduler implements the policy of Section~\ref{sec:aware-scheduling}. The path solves independent leaves in parallel and merges the tree bottom-up, using level-parallel work while enough merge pairs remain. NUMA placement keeps a level on one socket while its parallel work is too small to fill the node and widens it to both sockets only when enough merge pairs remain. Once the tree narrows, the scheduler probes a representative merge, reads its \(r\), and switches high-active levels to secular-root teams. These teams fold the secular solver's \(\delta\) values into thread-private WPROD accumulators and combine them before the streaming row update, with thread-indexed scratch sized to preserve root-level parallelism at high thread counts. These choices affect only how threads act on the shared boundary rows, not the serial merge semantics or the linear workspace bound of Section~\ref{sec:complexity}.

\subsubsection{GPU Implementation}

The GPU path uses one preallocated device workspace when tridiagonal eigenvectors are not requested. Its persistent arrays are all \(O(n)\), namely the two boundary rows, a small constant number of length-\(n\) split and deflation metadata arrays, and the compact secular-root scratch. The solver builds the split tree on the GPU and solves the leaves in one batched kernel that gives a warp-sized block to each leaf of at most \(32\) rows, keeping only their first and last eigenvector rows.

Merges then advance level by level, and all merges of a level are processed together so that every step is a single batched kernel over the level's index range rather than one launch per merge. Each level updates its split metadata, sorts the pole array, and deflates through a short sequence of kernels that flag negligible-component and close-pole candidates, compact the survivors, and rotate close pairs, all acting on the boundary rows rather than eigenvector blocks. The post-deflation active ranks feed the parallel policy of Section~\ref{sec:aware-scheduling}, that small-rank levels use batched kernels with lightweight per-root work, while high-rank levels dedicate a CUDA block to each active secular root, stage its denominator column in shared memory, and solve the secular equation by cooperative reduction. The corrected L\"owner update vector \(\widehat w\) is then formed by a WPROD accumulation over the compact-$\delta$ denominator columns and a finalize kernel, without a dense \(K\times K\) denominator table.

Boundary-row propagation runs one thread group per parent column, evaluating that column's secular entries from the compact-$\delta$ record of root origin, offset, and cached near-pole values and reducing them against the two selected child rows, the same arithmetic a block update would do one column at a time. At the root merge the kernels return the eigenvalues and skip propagation, so the GPU carries only boundary rows and compact secular-root records across levels instead of the dense eigenvector and replay state of the full-vector path.

\section{Experiments}
\label{sec:experiments}

% The preceding sections make three claims about BR: it should preserve the D\&C merge semantics, replace persistent values-only D\&C replay state with linear-size boundary data, and expose more parallelism than the low-memory QR/QL path when the D\&C split tree leaves useful parallel work.
This section tests those claims as one evidence chain. We first compare BR against the CPU QR/QL, MRRR eigenvalue-only, bisection, and internal values-only D\&C baselines, then trace CPU strong scaling, and finally test the same values-only D\&C design point against cuSOLVER on an H100.

\subsection{Experimental Setup}
\label{subsec:exp-setup}

We evaluate BR on a cluster using two execution environments.  CPU experiments use one dual-socket Intel Xeon Gold~6348 node with 56 physical cores.  The BR CPU code is compiled with \texttt{gfortran -O2} and linked against Intel oneAPI MKL~2025.3, and every baseline runs against that same library on the same node with identical thread pinning and MKL dynamic threading disabled, so BR and each reference share one software stack and its speedups reflect the algorithm rather than compiler tuning against the vendor-optimized routines it is compared with.  GPU experiments use one NVIDIA H100-PCIe with CUDA~13.2 and eight host cores.  Every reported time is the arithmetic mean of five timed runs and isolates the values-only tridiagonal divide-and-conquer stage, with dense reduction and back transformation outside the measured kernel and all computation in FP64 precision across BR and every CPU and GPU baseline.

The CPU suite uses four tridiagonal input families.  Uniform, normal, and clustered inputs are generated from prescribed spectra with \texttt{DLATMS}, while Toeplitz keeps a structured tridiagonal form with constant diagonal and off-diagonal entries.  Floating-point checks compare QR/QL, the internal values-only D\&C path, and BR against selected 256-bit MPFR Sturm-bisection references (Table~\ref{tab:accuracy-br}).  Build flags, MKL and OpenMP settings, matrix generators, timing protocol, and workspace-query rules are recorded in the supplemental material.  The BR solver, CPU and GPU benchmark harnesses, and matrix generators will be released as open source so that the reported measurements can be reproduced.

\begin{table}[t]
\caption{Low-workspace CPU baseline.  \(r=K/m\) is the per-merge active-rank ratio, where low \(r\) means most secular work deflates away; ratio columns report \texttt{DSTERF}/BR speedup.}
\label{tab:cpu-dsterf-vs-br}
\centering
\footnotesize
\setlength{\tabcolsep}{1.5pt}
\begin{tabular}{llcrrr}
\toprule
Input & \(N\) & \(r\) & \shortstack{\texttt{DSTERF}\\time} & \shortstack{\texttt{DSTERF}/BR\\1T} & \shortstack{\texttt{DSTERF}/BR\\56T} \\
\midrule
\multicolumn{6}{l}{\emph{Spectrum-prescribed}} \\
Uniform & 16{,}384 & 0.33 & 4.03~s & \(65.8\times\) & \(612.0\times\) \\
Uniform & 65{,}536 & 0.27 & 56.35~s & \(140.0\times\) & \(335.6\times\) \\
Normal & 16{,}384 & 0.32 & 3.50~s & \(58.0\times\) & \(505.2\times\) \\
Normal & 65{,}536 & 0.27 & 47.29~s & \(144.1\times\) & \(470.0\times\) \\
Clustered & 16{,}384 & 0.05 & 42.1~ms & \(2.39\times\) & \(10.30\times\) \\
Clustered & 65{,}536 & 0.06 & 257~ms & \(3.17\times\) & \(14.96\times\) \\
\midrule
\multicolumn{6}{l}{\emph{Structured tridiagonal}} \\
Toeplitz & 16{,}384 & 0.60 & 3.40~s & \(1.58\times\) & \(20.11\times\) \\
Toeplitz & 65{,}536 & 0.58 & 51.96~s & \(1.59\times\) & \(19.93\times\) \\
\midrule
\multicolumn{6}{l}{\emph{Real dense-reduction cache}} \\
Reduced dense & 49{,}152 & 0.95 & 29.85~s & \(0.95\times\) & \(9.14\times\) \\
\bottomrule
\end{tabular}
\end{table}

\subsection{Baselines}
\label{subsec:baselines}

We compare BR against measured CPU baselines that cover the standard values-only solver families.  LAPACK/MKL \texttt{DSTERF} is the standard eigenvalue-only QR/QL routine, whose serial sweep shows no 1-vs-56-thread MKL speedup, and \texttt{xSTEGR} with \texttt{JOBZ=N} is the eigenvalue-only MRRR wrapper.  MKL ScaLAPACK \texttt{PDSTEBZ} gives a parallel bisection baseline on 56 single-threaded MPI ranks.  The internal values-only D\&C path, MKL \texttt{DLAED0(ICOMPQ=0)}, is the closest algorithmic comparison because it runs the same D\&C secular merges and returns eigenvalues only, yet carries replay state and therefore much larger workspace; we measure this path through \(N=16{,}384\) under LP64 and at selected larger rows under ILP64.  ELPA is not a values-only baseline because its public \texttt{solve\_tridiagonal} interface computes and stores eigenvectors.  On the GPU we compare against cuSOLVER \texttt{cusolverDnXstedc} with \texttt{compz=N}, the public cuSOLVER tridiagonal D\&C interface for eigenvalues only.

\subsection{Comparison with CPU Baselines}
\label{subsec:cpu-baselines}

BR sits against CPU references that answer different questions.  The internal values-only D\&C path is the like-for-like algorithmic comparison, returning the same eigenvalues from the same secular-merge family and differing from BR only in the state it carries across levels.  QR/QL (\texttt{DSTERF}) is the default-path reference, since a no-vector \texttt{stedc} request falls back to it rather than to D\&C, so it is the routine users actually reach today.  \texttt{xSTEGR(JOBZ=N)} and \texttt{PDSTEBZ} add measured non-D\&C baselines for the MRRR eigenvalue-only wrapper and parallel bisection, which we keep in Table~\ref{tab:complexity} and the supplement rather than the main table because they stay far from BR's range while QR/QL is the path users actually reach.

Against the like-for-like path, replacing replay state with boundary rows is a clean win at identical semantics (Table~\ref{tab:internal-dc-vs-br}).  Single-threaded BR is \(2.8\)--\(3.9\times\) faster on uniform, normal, and clustered inputs, and at 56 threads the gain reaches \(14.6\times\) on Toeplitz, where enough secular work survives deflation to reward the full node.  It wins on memory by a wider margin than on time: at \(N=16{,}384\) the internal path already queries about 6.0~GiB while BR uses a few MiB, and by \(N=65{,}536\) the internal path's quadratic state is out of memory on the same node (Table~\ref{tab:complexity}), so the only other D\&C competitor cannot even run at the scale BR is built for.

Because the default no-vector path is QR/QL rather than D\&C, Table~\ref{tab:cpu-dsterf-vs-br} also reports BR against \texttt{DSTERF}, with the one-thread ratio isolating the algorithmic gain and the 56-thread ratio adding the schedule.  These ratios are asymptotic rather than tuning artifacts, since \texttt{DSTERF} does \(\Theta(n^2)\) work on every spectrum whereas deflation lets BR solve only the \(K\) roots that survive, so a hard-deflating input leaves QR/QL grinding through work BR never touches.  Uniform and normal spectra deflate hard while QR/QL stays slow, so one-thread BR is already \(>50\times\) faster, while Toeplitz and the real dense-reduction cache keep large active fractions, run near QR/QL at one thread, and take their \(9\)--\(20\times\) advantage from the node schedule.  Clustered is the opposite corner, where BR deflates about \(95\%\) of the work but \texttt{DSTERF} also block-splits, so both finish in a few milliseconds and the ratio stays modest by construction.  Section~\ref{subsec:cpu-strong-scaling} traces the full thread sweep, and the per-merge active-rank profile behind these ratios, a QR-versus-internal-D\&C size sweep, and further MRRR and ELPA probes are reported in the supplemental material.

The measured MRRR and bisection baselines in Table~\ref{tab:complexity} show the same separation from BR at the large uniform case.  The \texttt{xSTEGR(JOBZ=N)} wrapper takes 88.85~s and shows essentially no MKL thread speedup in our 1-vs-56-thread check, while 56-rank \texttt{PDSTEBZ} reduces the bisection time to 18.34~s.  Both remain far above BR's 168~ms because they do not retain the D\&C merge structure whose surviving secular work BR schedules directly.

\begin{table}[t]
\caption{Internal values-only D\&C comparison with matched input families, with the one-thread and 56-thread measurements placed side by side.  Times are elapsed; ratios above one favor BR.}
\label{tab:internal-dc-vs-br}
\centering
\footnotesize
\setlength{\tabcolsep}{2.4pt}
\begin{tabular}{@{}llcccccc@{}}
\toprule
 & & \multicolumn{3}{c}{1 thread} & \multicolumn{3}{c}{56 threads} \\
\cmidrule(lr){3-5}\cmidrule(lr){6-8}
Input & \(N\) & Int. & BR & Int./BR & Int. & BR & Int./BR \\
\midrule
Uniform & 16{,}384 & 170~ms & 61.2~ms & \(2.77\times\) & 17.5~ms & 6.56~ms & \(2.66\times\) \\
Normal & 16{,}384 & 169~ms & 60.3~ms & \(2.80\times\) & 17.0~ms & 6.92~ms & \(2.46\times\) \\
Toeplitz & 16{,}384 & 2.60~s & 2.16~s & \(1.21\times\) & 2.47~s & 169~ms & \(14.61\times\) \\
Clustered & 16{,}384 & 69.6~ms & 17.7~ms & \(3.94\times\) & 6.71~ms & 4.09~ms & \(1.64\times\) \\
\bottomrule
\end{tabular}
\end{table}

\subsection{Application-Shaped Tridiagonal Inputs}
\label{subsec:application}

The four families above come from LAPACK's generators, but the same values-only tridiagonal kernel arises in applications.  Gaussian quadrature nodes are eigenvalues of a symmetric tridiagonal Jacobi matrix~\cite{golubWelsch1969}, one-dimensional finite-difference and Schr\"odinger operators are tridiagonal by construction, and birth-death generators symmetrize into the same Jacobi form~\cite{teschl2000jacobi,karlin1957classification}, each entering BR as a real \(D,E\) pair.  Table~\ref{tab:application} reports three such spectra spanning the deflation regimes at \(N=16{,}384\) on 56 threads.

The speedup tracks how much the split tree deflates.  Gauss--Legendre barely deflates, keeping every root active at the final merge, yet the parallel merge alone already gives \(11.0\times\) over \texttt{DSTERF}; the 1D Laplacian halves the root merge for \(20.5\times\); and the disordered Anderson Schr\"odinger operator~\cite{anderson1958absence,shklovskii1993}, whose localized states make the rank-one couplings negligible, collapses the final merge from \(16{,}384\) to \(234\) active roots, giving \(523.6\times\) at \(N=16{,}384\) and \(2868\times\) at \(N=65{,}536\).  The large gains thus read directly off how strongly the physical spectrum deflates the upper merges.

\begin{table}[t]
\caption{Application-shaped tridiagonal inputs at \(N=16{,}384\) on 56 threads.  \(r\) is the merge-size-weighted active-rank ratio and root \(K\) the active rank at the final merge; the last column is the \texttt{DSTERF}/BR speedup.}
\label{tab:application}
\centering
\footnotesize
\setlength{\tabcolsep}{4pt}
\begin{tabular}{@{}lccrr@{}}
\toprule
Input & Source & \(r\) & root \(K\) & \texttt{DSTERF}/BR \\
\midrule
Gauss--Legendre Jacobi & quadrature & 0.99 & 16{,}384 & \(11.0\times\) \\
1D Laplacian & PDE operator & 0.60 & 8{,}192 & \(20.5\times\) \\
Anderson Schr\"odinger & disordered operator & 0.47 & 234 & \(523.6\times\) \\
\bottomrule
\end{tabular}
\end{table}

\subsection{CPU Strong Scaling}
\label{subsec:cpu-strong-scaling}

Tables~\ref{tab:cpu-dsterf-vs-br} and~\ref{tab:internal-dc-vs-br} report endpoint timings.  Figure~\ref{fig:cpu-strong-scaling} gives the corresponding BR self-speedup sweep \(S_p=T_1/T_p\) at \(N=65{,}536\).

\begin{figure}[t]
    \centering
    \includegraphics[width=0.85\columnwidth]{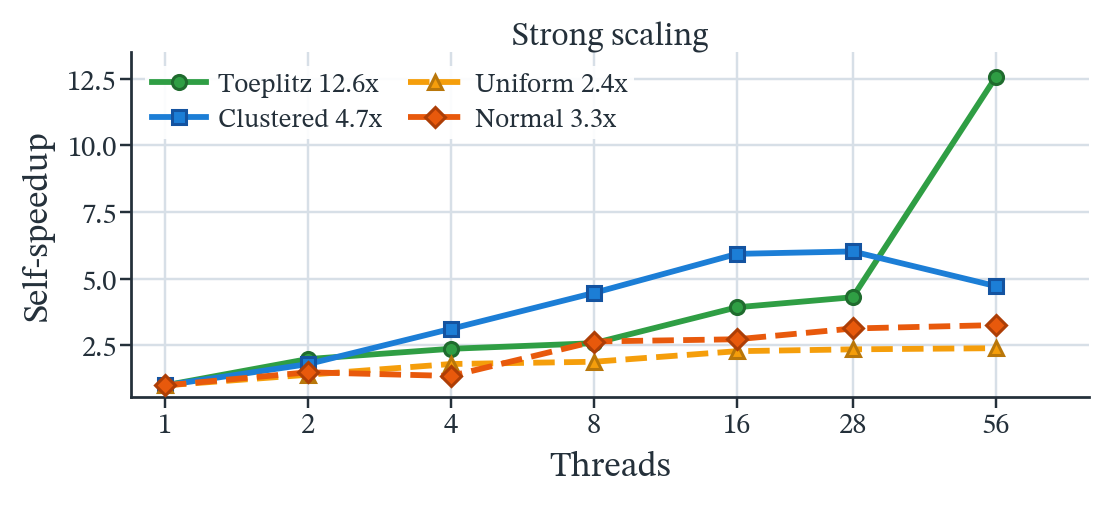}
    \Description{CPU strong-scaling plot at N=65536, with endpoint speedup labels for Toeplitz, clustered, uniform, and normal inputs.}
    \caption{BR self-speedup at \(N=65{,}536\) by input family.}
    \label{fig:cpu-strong-scaling}
\end{figure}

The sweep supports the endpoint trends without restating the table values.  Toeplitz keeps substantial active secular work and therefore benefits most from the full node.  Uniform and normal deflate more aggressively, so their one-thread times are already low and the remaining parallel speedup is modest.  Clustered is the most deflation-heavy row; its best time occurs before 56 threads because the residual secular work is too small to amortize full-node scheduling overhead.

Figure~\ref{fig:scheduler-ablation} isolates the scheduler choice, and no static policy wins across the deflation regimes.  Secular-root parallelism runs \(2.12\)--\(3.16\times\) slower than adaptive on uniform, normal, and clustered, whose merges deflate to little surviving work, while task-level parallelism runs \(4.79\times\) slower on Toeplitz, whose narrow upper levels keep enough active secular work to fill the node.  Reading the active rank \(r\) each merge keeps, the adaptive policy tracks the better static choice in every family and pays only a \(2\)--\(3\%\) probe overhead where task-level already suffices.  The internal D\&C panel separates state from schedule on the \(N=49{,}152\) reduced-dense row, where replacing replay state with BR already gives \(6.6\times\) and the adaptive scheduler raises the gain to \(21.4\times\).

\begin{figure}[t]
    \centering
    \includegraphics[width=\columnwidth]{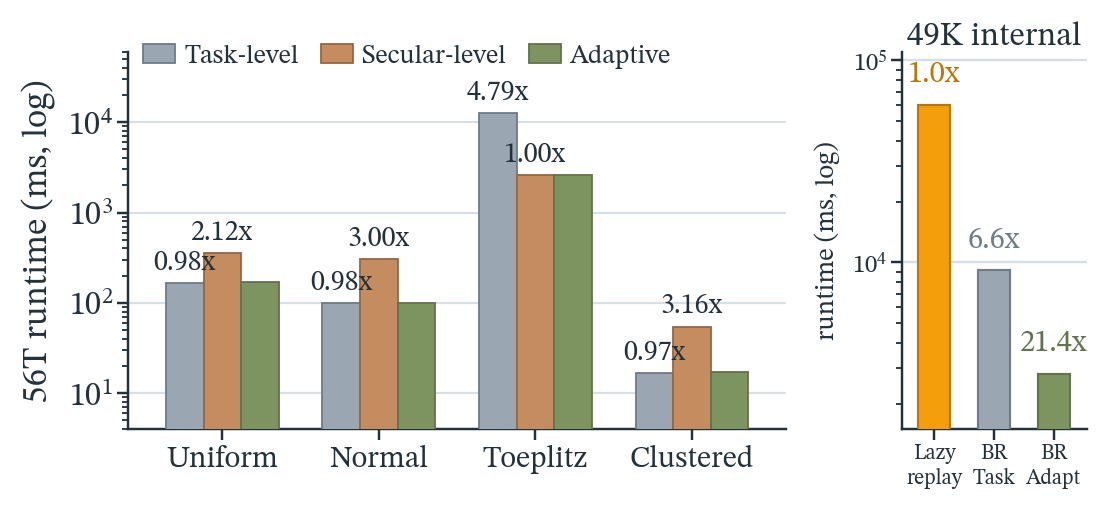}
    \Description{Grouped bar chart at N=65536 comparing task-level, secular-level, and adaptive 56-thread runtime on a log axis for uniform, normal, Toeplitz, and clustered inputs, plus a N=49152 internal ablation with lazy replay, BR task-level, and BR adaptive runtime bars.}
    \caption{Scheduler ablation and state/schedule decomposition.  Left: \(N=65{,}536\) BR scheduler variants on 56 threads, with labels reporting static/adaptive runtime ratios.  Right: \(N=49{,}152\) internal D\&C ablation with lazy replay, BR task-level, and BR adaptive runtime bars; labels report speedup over lazy replay.}
    \label{fig:scheduler-ablation}
\end{figure}

\subsection{GPU D\&C Against cuSOLVER}

The H100 experiment tests the same reduction on an accelerator, where linear state should both speed up the D\&C merge and let it run past the memory ceiling that cuSOLVER's quadratic workspace imposes.  Figure~\ref{fig:h100-runtime-dc} compares our GPU EVD prototype's values-only D\&C path with cuSOLVER \texttt{Xstedc} on the same tridiagonal inputs.  The comparison is intentionally DC-only, with both paths starting from tridiagonal matrices, computing eigenvalues only, and excluding dense reduction and back transformation.

The results confirm both halves.  Where both paths fit in device memory, BR is faster on every row, with cuSOLVER/BR speedups ranging from \(1.14\times\) to \(25.82\times\) through \(N=32{,}768\), and \(3.43\)--\(11.39\times\) at \(N=49{,}152\).  The largest occur on Toeplitz, where the secular work remains substantial; we also observe and analyze a cuSOLVER-only nonmonotonicity in this row.\footnote{A kernel profile points to size-dependent merge granularity, that the \(N=32{,}768\) run is dominated by a single-block root \texttt{laed2\_par} merge, whereas at \(N=49{,}152\) the work shifts toward multi-block merge/update kernels.}  Beyond that size the gap becomes categorical. cuSOLVER's queried workspace reaches 96.0~GiB at \(N=65{,}536\), exceeding the 80~GB H100, while BR returns the uniform row in 60.5~ms.  On the GPU the quadratic state, not throughput, is what ends cuSOLVER at scale, the ceiling the boundary-row reduction removes.

\begin{figure}[t]
    \centering
    \includegraphics[width=0.88\columnwidth]{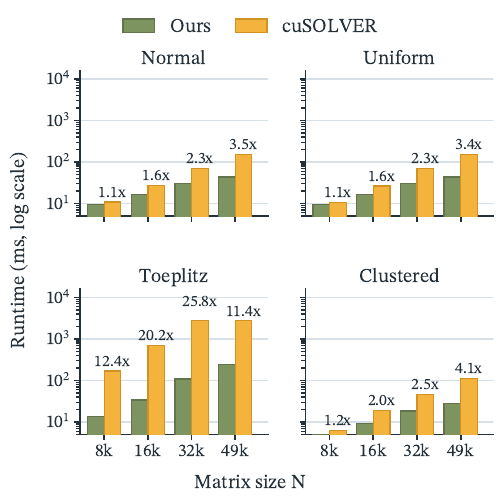}
    \caption{GPU tridiagonal D\&C comparison.  Labels report cuSOLVER/BR speedup.}
    \Description{Four-panel H100 runtime bar chart for normal, uniform, Toeplitz, and clustered tridiagonal inputs at N=8192, 16384, 32768, and 49152. Each panel compares our GPU EVD prototype against cuSOLVER Xstedc, with speedup labels above the cuSOLVER bars.}
    \label{fig:h100-runtime-dc}
\end{figure}

\section{Conclusion}
\label{sec:conclusion}

This paper presents boundary-row D\&C, a values-only tridiagonal eigensolver that keeps the parallel D\&C merge structure while replacing dense replay state with propagated boundary rows.  The boundary state constructs the same secular problems as conventional D\&C in exact arithmetic, reduces persistent auxiliary storage from quadratic to linear, and maps to both CPU and GPU prototypes.  BR is not a universal replacement for QR/QL, but it makes D\&C-style parallelism available to eigenvalue-only workloads at a spectrum-dependent gain that is widest on well-spread spectra where BR deflates while QR/QL stays slow, and narrower both where the QR/QL baseline is itself cheap, as on highly clustered spectra, and where the residual secular work stays large, as on Toeplitz.

The linear boundary-row state also points to a distributed solver as the main direction for future work.  ScaLAPACK \texttt{PxSTEDC} is communication-bound because it redistributes $O(n^2)$ eigenvector state at every level, whereas a values-only BR merge needs only the two boundary rows of each child, communicating $O(n)$ data per merge.  Turning that reduction into a matching communication reduction is the natural next step toward a distributed values-only solver.

\bibliographystyle{ACM-Reference-Format}
\bibliography{ref}

\end{document}